\documentclass[a4paper]{amsart}

\usepackage{amsmath,amssymb,enumerate}
\usepackage{graphicx}
\usepackage{a4wide}
\usepackage{hyperref}

\usepackage[subrefformat=parens,labelformat=parens,labelfont=rm]{subfig}

\newtheorem{theorem}{Theorem}
\newtheorem{lemma}{Lemma}
\newtheorem{prop}{Proposition}

\newcommand{\doi}[1]{\href{http://dx.doi.org/#1}{\texttt{doi:#1}}}

\makeatletter
\@namedef{subjclassname@2020}{%
  \textup{2020} Mathematics Subject Classification}
\makeatother

\begin{document} 

\title{Paired 2-disjoint path covers of burnt pancake graphs with faulty elements} 

\author{Tom\'a\v{s} Dvo\v{r}\'ak \and Mei-Mei Gu}
\address{Faculty of Mathematics and Physics, Charles University, Prague, Czech Republic}
\email{dvorak@ksvi.mff.cuni.cz, mei@kam.mff.cuni.cz}

\begin{abstract} 
The burnt pancake graph $BP_n$ is the Cayley graph of the hyperoctahedral group using prefix reversals as generators.
Let $\{u,v\}$ and $\{x,y\}$ be any two pairs of distinct vertices of $BP_n$ for  $n\geq 4$. We show that there are $u-v$ and $x-y$ paths whose vertices partition the vertex set of $BP_n$ even if $BP_n$ has up to $n-4$ faulty elements. On the other hand, for every $n\ge3$ there is a~set of $n-2$ faulty edges or faulty vertices for which such a fault-free disjoint path cover does not exist.
\end{abstract} 



\maketitle

\thispagestyle{empty}
\section{Introduction}

Consider a stack of pancakes of different diameters where each pancake has a burnt side. A~chef spatula can be inserted at any point in the stack and used to flip all pancakes above it. The burnt pancake problem, introduced by Gates and Papadimitriou \cite{Gates-1979}, is to sort the stack in order of size using the minimum number of flips so that all pancakes end up with the burnt side on bottom.

The problem may be modeled by a \emph{burnt pancake graph} $BP_n$ whose vertices and edges represent the stacks of $n$ burnt pancakes and their flips, respectively. As the graphs defined in this way are actually Cayley graphs of the group of signed permutations generated by prefix reversals, they possess a~number of useful properties such as vertex-transitivity or relatively small degree and diameter with respect to the size of the graph. Moreover, they are relatively sparse, compared to e.g.~the class of hypercubes, which makes them appealing as a~model of interconnection networks for parallel and distributed computing.     
This inspired an extensive  research of burnt pancake graphs focused on their structural properties as well as related algorithmic problems such as  matching preclusion \cite{hu}, routing schemes \cite{Iwasaki-2010}, fault tolerance \cite{So15}, conditional diagnosability \cite{So16}, component connectivity \cite{gc}, weak pancyclicity \cite{bbp} or neighbor connectivity \cite{gc22}.

A \emph{$k$-disjoint path cover} ($k$-DPC) of a graph $G$ is a~set of $k$ paths whose vertex sets form a~partition of the vertices of $G$. This concept originated from the interconnection networks community and was motivated by applications where the full utilization of network nodes is important \cite{Park06}.
The problem of construction of a $k$-DPC joining given sets of $k$ sources
and $k$  sinks has been studied in two incarnations: In the \emph{paired} version, the $i$-th path runs between the $i$-th source and $i$-th sink while in the \emph{unpaired} version each path joins an arbitrary source and sink. The problem was previously investigated for general graphs \cite{Enomoto00} as well as for various network models including recursive circulants \cite{Kim13}, grid graphs \cite{Park15}, hypercube-like networks \cite{Park06}, faulty hypercubes \cite{Dvorak08}, hypercubes \cite{Dvorak17}, interval graphs \cite{ParkLim22} or torus-like graphs \cite{Park22}. 

In this paper we study paired $2$-disjoint path covers of burnt pancake graphs with \emph{faulty elements} representing network nodes or communication links that have become overloaded or unavailable.
The paper is organized as follows. 
The next section introduces terminology, notation and surveys previous results on burnt pancake graphs. Section \ref{sec:tools} provides tools that are used as building blocks for the inductive construction in the proof of the main theorem, described in Section~\ref{sec:main}. The paper is concluded with a discussion on the optimality of the main result and an open problem that may serve as inspiration for future research.
\section{Preliminaries}

The aim of this section is to introduce notation and terminology. 
In the rest of this paper, $n$ always stands for a positive integer while $\bar{k}$ denotes the negation of an integer $k$, i.e., $\bar{k}=-k$.
We use $[n]$ and $\langle n\rangle$ to denote the sets $\{1,2,\dots, n\}$ and $[n]\cup\{ \bar k \mid k\in [n]\}$, respectively.

As usual, $V(G)$ stands for the vertex set and $E(G)$ for the edge set of a graph $G$. The distance of vertices $u,v$ of $G$ is denoted by $d_G(u,v)$.
Given the subgraphs $H_1,H_2,\dots, H_k$ of $G$, $\bigcup_{i=1}^k H_i$ denotes the subgraph of $G$ induced by the vertices of $\bigcup_{i=1}^k V(H_i)$.

Let $F\subseteq V(G)\cup E(G)$ be a set of \emph{faulty elements} of $G$ which may include both \emph{faulty vertices} and \emph{faulty edges}. Then $G-F$ stands for the graph with vertices $V(G)\setminus F$ and edges $E(G)\setminus (F \cup E')$ where $E'$ is the set of all edges of $G$ incident with the vertices of $F$. Vertices (edges) of $G-F$ are then called \emph{fault-free vertices} (\emph{fault-free edges}) of $G$. To simplify the notation, in the case that $F=\{x\}$ we use $G-x$ to denote $G-\{x\}$. 
Furthermore, we use $G-F_1-F_2$ to denote the graph $(G-F_1)-F_2$.
If $H$ is a subgraph of $G$, $G-H$ denotes the graph $G-V(H)$.   

\subsection{Paths and cycles}

A~sequence $\langle u_1,u_2,\dots,u_{n}\rangle$ of distinct vertices such that any two consecutive vertices are adjacent is a~\emph{path} between vertices $u_1$ and $u_n$. The vertex and edge sets $\{u_1, u_2, \dots,u_{n}\}$ and $\{u_1u_2, u_2u_3, \dots,u_{n-1}u_n\}$ of such a path $P$ are denoted by $V(P)$ and $E(P)$, respectively. The length of a path $P$ is defined as the size of $E(P)$ and denoted by $\ell(P)$. 
For the notational convenience, a~path $P$ between vertices $u$ and $v$ is denoted by $P[u,v]$. Given a path $P[u,v]$, we use $P[v,u]$ to denote the path between $v$ and $u$ obtained by reversing the sequence $P[u,v]$.    

Let $P=\langle u_1, u_2,\dots,u_{n}\rangle$ and $Q=\langle u_n, u_{n+1}, \dots,u_{m}\rangle$ be paths such that $V(P)\, \cap\, V(Q)= \{u_n\}$. Then the \emph{concatenation} of $P$ with $Q$, written as $P+Q$, is the path $\langle u_1$, $u_2, \dots,u_{n},u_{n+1},\dots,u_m\rangle$. As the operation + is associative, we can safely use $\sum_{i=1}^kP_i$ instead of $P_1+P_2+\cdots+P_k$.

If $\langle u_1$,$u_2$,$\dots,u_{n}\rangle$ is a path and $u_n$ is adjacent to $u_1$, the sequence  $\langle u_1$,$u_2$,$\dots,u_{n},u_1\rangle$ is a \emph{cycle} of length $n$. The vertex and edge sets $\{u_1, u_2, \dots,u_{n}\}$ and $\{u_1u_2, u_2u_3, \dots,u_{n-1}u_n, u_nu_1\}$ of such a cycle $C$ are denoted by $V(C)$ and $E(C)$, respectively, while its length is denoted by $\ell(C)$. Note that therefore we can view $C$ as a~graph with vertices $V(C)$ and edges $E(C)$.
A \emph{Hamiltonian cycle} (\emph{Hamiltonian path}) of a~graph $G$ is a cycle (path) that contains each vertex of $G$ exactly once. A~graph is \emph{Hamiltonian} if it contains a Hamiltonian cycle. A~graph is \emph{Hamiltonian-connected} if there is a Hamiltonian path between any pair of distinct vertices of the graph.

Let $C=\langle u_1$,$u_2$,$\dots,u_{n},u_1\rangle$ be a cycle, $i\ne j \in[n]$ and put $u_0=u_n, u_{n+1}=u_1$. Vertices $a$ and $b$ are called 
\begin{itemize}
  \item $\emph{concordant neighbors}$ of $u_i$ and $u_j$ on $C$ if $(a,b)\in\{(u_{i-1},u_{j-1}),(u_{i+1},u_{j+1})\}$,
  \item $\emph{discordant neighbors}$ of $u_i$ and $u_j$ on $C$ if $(a,b)\in\{(u_{i-1},u_{j+1}),(u_{i+1},u_{j-1})\}$.
\end{itemize}

Let $S=\{s_1,s_2,\ldots,s_k\}$ and $T=\{t_1,t_2,\ldots,t_k\}$ be two disjoint subsets of distinct vertices of a~graph $G$.
A~\emph{paired $k$-disjoint path cover} ($k$-DPC for short) \emph{of $G$ joining $S$ and $T$} is a set $\{P_1,P_2,\ldots,P_k\}$ of $k$ pairwise vertex-disjoint paths in $G$ such that $P_i=P[s_i,t_i]$ for all $i\in[k]$
and $\bigcup_{i=1}^{k} V(P_i)=V(G)$. Sets $S$ and $T$ are called \emph{terminal sets of size $k$}, their elements are \emph{terminal vertices}. In the following we always assume that terminal sets are disjoint and of the same size.

Note that if a graph on at least four vertices admits a 2-DPC for arbitrary terminal sets, then it must be Hamiltonian-connected.

\subsection{Burnt pancake graphs}

A \emph{signed permutation} of $[n]$ is an $n$-permutation $u_1u_2\cdots u_n$ of $\langle n\rangle$ such that $|u_1||u_2|\cdots |u_n|$ is a~permutation of $[n]$. Given an integer $i\in[n]$ and a signed permutation $u=u_1u_2\cdots u_i\cdots u_n$ of $\langle n \rangle$, the \emph{$i$-th prefix reversal} of $u$ is denoted by $u^i$ and defined as the signed permutation $\bar{u}_i\bar{u}_{i-1}\cdots \bar{u}_1u_{i+1}\cdots u_n$. As an example consider a~signed permutation $u=1\bar{2}3\bar{5}4$ of $[5]$, then
$u^3=\bar{3}2\bar{1}\bar{5}4$ while $u^5=\bar{4}5\bar{3}2\bar{1}$.

The \emph{$n$-dimensional burnt pancake graph}, denoted by $BP_n$, is defined \cite{Gates-1979} as the graph whose vertex set consists of all  signed permutations of $\langle n\rangle$, two vertices $u,v\in V(BP_n)$ being adjacent whenever $u^i=v$ for some $i\in[n]$. 
Fig.~\ref{fig:BP-123} depicts $BP_n$ for all $n\in[3]$. Note that $BP_2$ is an 8-cycle while the girth of $BP_n$ for $n\ge2$ is known to be equal to 8 \cite{Compeau-2011}. 

The definition implies that $BP_n$ is an $n$-regular graph, $|V(BP_n)|=n!\cdot 2^{n}$ and $|E(BP_n)|=n\cdot n!\cdot 2^{n-1}$. It is easy to see that $BP_n$ is not a bipartite graph for any $n\ge3$: Indeed, $BP_3$ contains a~9-cycle
\[
\langle 123,\bar2\bar13,\bar{3}12,
312,\bar1\bar32,\bar{2}31,231,
\bar3\bar21,\bar{1}23,
123\rangle
\] 
and the same is true for any $BP_n$ with $n\ge4$ since it contains $BP_3$ as a subgraph as explained below.

\begin{figure}[htb]
\begin{center}
\includegraphics{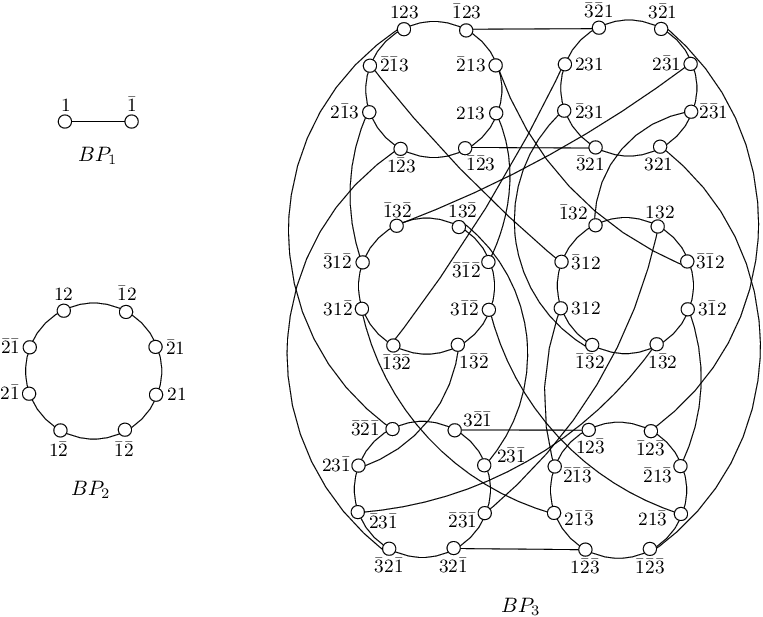}
\caption{Burnt pancake graphs $BP_n$ for $n\in[3]$.}
\label{fig:BP-123}
\end{center}
\end{figure}

$BP_n$ can be partitioned into $2n$ vertex-disjoint subgraphs $BP_n^i$,  $i\in\langle n \rangle$, where $BP_n^i$ is the subgraph induced by the vertices $\{u_1u_2\cdots u_n\in V(BP_n)\mid u_n=i\}$. 
Note that each $BP_n^i$ is isomorphic to $BP_{n-1}$. Given distinct $i,j\in\langle n\rangle$, the set of all edges between $BP_{n}^{i}$ and $BP_{n}^{j}$ is denoted by $E_{i,j}(BP_n)$. Note that the edges of $E_{i,j}(BP_n)$ always form a matching. 
An edge of $BP_n$ is called an \emph{out-edge} if is falls into $E_{i,j}(BP_n)$ for some $i,j\in\langle n\rangle$.
Given a~vertex $v\in V(BP_n^i)$, a unique out-edge leads from $v$ to its neighbor outside of  $BP_n^i$, denoted by $v^n$ and called the \emph{out-neighbor} of $v$.

\subsection{Previous results}

The next two lemmas summarize useful properties of burnt pancake graphs that follow directly from the definition.

\begin{lemma}[\cite{Chin-2009,Compeau-2011,Iwasaki-2010}]\label{BPn}
Let $n\ge2$ and $i,j\in\langle n\rangle$ such that $i\ne j$. Then
\[
|E_{i,j}(BP_n)|=
\begin{cases}
(n-2)!\cdot 2^{n-2} &\text{ for } i\neq\overline{j}\\ 
0 &\text{ otherwise\,} .
\end{cases}
\]
\end{lemma}

\begin{lemma}[\cite{hu}]\label{BPn0}
Let $n\ge3$, $u\in V(BP_n^i)$ and $v\in V(BP_n^j)$ for some $i,j\in \langle n\rangle$.
\begin{enumerate}[\upshape(1)]
\item If $i=j$ and $1\leq d_{BP_n}(u,v)\leq 2$, then $u^n\neq v^n$.
\item If $i\neq j$ and $d_{BP_n}(u,v)\leq 3$, then $u^n\neq v^n$.
\end{enumerate}
\end{lemma}

The following results on Hamiltonicity of burnt pancake graphs with faulty elements are due to Kaneko \cite{kk}.

\begin{lemma}[\cite{kk}]\label{BPn2}
 Let $n\geq 3$ and $F\subseteq V(BP_n)\cup E(BP_n)$. If $|F|\leq n-2$, then $BP_n-F$ is Hamiltonian. If $|F|\leq n-3$, then  $BP_n$ is Hamiltonian-connected.
\end{lemma}

\begin{lemma}[\cite{kk}]\label{BPn1}
Let $F\subseteq V(BP_n)\cup E(BP_n)$ with $|F|\leq n-2$ and $\{k_1,k_2,\ldots,k_m\}\subseteq \langle n\rangle$ such that $n\geq 4$ and $m\geq 5$.
If each of the $m$ subgraphs $BP_n^{k_1}-F$, $BP_n^{k_2}-F$,
$\ldots$, $BP_n^{k_m}-F$ is Hamiltonian-connected, then for each $u\in V(BP_n^{k_i})$
and $v\in V(BP_n^{k_j})$, $1\leq i\neq j\leq m$, there is a~Hamiltonian path between $u$ and $v$ in $\bigcup_{i=1}^mBP_i^{k_i} -F$.
\end{lemma}

\section{Tools}
\label{sec:tools}
In this section we describe the tools that are necessary for our constructions of Hamiltonian paths and disjoint path covers. We start with two simple  observations.
\begin{prop}\label{prop:edge-sel}
Let $n\ge4$, $F$ be a set of at most $n-3$ faulty elements in $BP_n$, and $k_1\ne k_2\in\langle n\rangle$.
\begin{enumerate}[\upshape(1)]
\item\label{prop:edge-sel:1}
If $BP_n^{k_1}-F$ admits a disjoint path cover,  $x,y,z,w$ are arbitrary distinct vertices of $BP_n-F$ and $k_2\ne\overline{k_1}$, then some of its paths contains an edge $ab$ such that $a^n\in V(BP_n^{k_2}-F)$, $b^n\in V(BP_n^{k_3}-F)$ for some $k_3\in\langle n\rangle\setminus\{k_1,k_2,\overline{k_1}\}$,
$\{a,b,a^n,b^n\}\cap\{x,y,z,w\}=\emptyset$ and both out-edges $aa^n,bb^n$ are fault-free.
\item\label{prop:edge-sel:2} 
If $BP_n^{k_1}-F$ contains a Hamiltonian cycle $C$ and $x,y$ are arbitrary distinct vertices of $BP_n-F$, then $C$ contains an edge $ab$ such that $a^n,b^n$ fall into $BP_n^{k_3}-F, BP_n^{k_4}-F$, respectively, for some $k_3,k_4\in\langle n \rangle\setminus\{k_1,k_2,\overline{k_1}\}$ such that $k_3\ne k_4$, $\{a,b,a^n,b^n\}\cap\{x,y\}=\emptyset$ and both out-edges $aa^n,bb^n$ are fault-free.
\item\label{prop:edge-sel:3} 
If $x,y,z$ are arbitrary distinct vertices of $BP^{k_1}_n-F$, $e\in F$ and $k_2\ne\overline{k_1}$, then there is a vertex $a$ in $BP^{k_1}_n-F-\{x,y,z\}$ such that  $a^n\in V(BP_n^{k_2}-F)$ and the out-edge $aa^n$ is fault-free.
Moreover, if $e$ is an edge, then $a$ is not incident with $e$, if $e$ is a vertex, then $a$ is not its neighbor.
\end{enumerate}
\end{prop}
\begin{proof}
\eqref{prop:edge-sel:1}
As $k_2\ne\overline{k_1}$, Lemma~\ref{BPn} implies that $BP_n^{k_1}$ contains $(n-2)!\cdot 2^{n-2}$ vertices $a_i$ such that $a_i^n\in BP_n^{k_2}$. If $a_i$ is fault-free, any DPC of $BP_n^{k_1}$ includes a path $P$ that passes through $a_i$. Since $\ell(P)\ge1$, there is a neighbor $b_i$ of $a_i$ on $P$. Note that if $a_i,a_j$ are two such distinct vertices, then Lemma~\ref{BPn0} implies $d(a_i,a_j)>2$ which means that their corresponding neighbors $b_i,b_j$ must be distinct as well. Each such edge $a_ib_i$ may be ``blocked'' by a faulty element of $F$ or a vertex of $\{x,y,z,w\}$. As $|F\cup\{x,y,z,w\}|\le n+1 < (n-2)!\cdot 2^{n-2}$ for $n\ge4$, there is an edge $ab$ such that both out-edges and out-neighbors of $a$ and $b$ are fault-free while $\{a,b,a^n,b^n\}\cap\{x,y,z,w\}=\emptyset$. By Lemma~\ref{BPn0}, the out-neighbors of $a$ and $b$ fall into distinct subgraphs, which means that $b^n\in V(BP_n^{k_3})$ for some $k_3\not\in\{k_1,k_2, \overline{k_1}\}$.

\eqref{prop:edge-sel:2}
Observe that any four consecutive vertices on $C$ induce a path of length three which --- by Lemma~\ref{BPn0} --- contains an edge $ab$ such that neither $a^n$ nor $b^n$ falls into $BP_n^{k_2}$.
There are possibly $\lfloor |V(C)|/4 \rfloor$  pairwise vertex-disjoint paths of length three on $C$ with such an edge $ab$, but
\begin{itemize}
\item at most two of them may have its vertices or their out-neighbors in $\{x,y\}$,
\item at most $f_2$ of them may contain vertex with a faulty out-edge  or faulty out-neighbor,
\item where $|V(C)|=|V(BP_n^{k_1})|-f_1$ while $f_1+f_2=|F|\le n-3$.
\end{itemize}
 Since
 \[
\left\lfloor \frac{|V(BP_n^{k_1})|-f_1}{4}\right\rfloor-f_2-2\ge(n-1)!\cdot 2^{n-3}-(n-3)-2>1
 \]
for $n\ge4$, the cycle $C$ contains an edge $ab$ with the desired properties. The fact that $a^n,b^n$ fall into distinct subgraphs follows from Lemma~\ref{BPn0}.

\eqref{prop:edge-sel:3}
By Lemma~\ref{BPn} there are $(n-2)!\cdot2^{n-2}$ vertices of $BP_n^{k_1}$ whose out-neighbors lie in $BP_n^{k_2}$. Note that
\begin{itemize}
\item at most $f_1\le3$ of them may fall into $\{x,y,z\}$,
\item at most $f_2\le|F|\le n-3$ of them may be either faulty or incident with a faulty out-edge or adjacent to a~faulty out-neighbor,
\item if $e$ is an edge, at most $f_3\le2$ of them may be incident with it,
\item if $e$ is a vertex, at most $f_4\le n-1$ of them may be its neighbors.
\end{itemize}
As $f_1+f_2+\max\{f_3,f_4\}\le 2n-1 < (n-2)!\cdot2^{n-2}$ for $n\ge4$, a vertex $a$ with the desired properties exists.
\end{proof}
\begin{prop}\label{prop}
Let $\{k_1,k_2,\ldots, k_m\}\subseteq \langle n\rangle$ with $m\geq 5$. Then there is a permutation
$k_1',k_2',\cdots, k_m'$ of $k_1,k_2,\cdots, k_m$ such that $k_1=k_1'$, $k_m'=k_m$ and  $k_i'\neq \overline{k}_{i+1}'$ for $i\in[m-1]$.	
\end{prop}
\begin{proof}
Let $G$ be a graph on the vertex set $\langle n\rangle$, an edge joining two vertices $i,j\in\langle n\rangle$ whenever $i\ne\bar j$. Consider a subgraph $H$ of $G$ induced by the set $\{k_1,k_2,\ldots, k_m\}$. Then for every vertex $v$  of $H$ we have $d_H(v)\ge m-2>|V(H)|/2$ for $m\ge5$. Hence by  \cite[Problem~10.24]{Lo}, $H$ is Hamiltonian-connected. It follows that $H$ contains a Hamiltonian path  $k_1=k_1',k_2',\cdots, k_m'=k_m$ which forms the permutation with the desired properties. 	
\end{proof}
The next lemma extends Lemma~\ref{BPn1} to Hamiltonian paths between arbitrary endvertices provided that the number of subgraphs is larger.

\begin{lemma}\label{lem:ham-conn}
Let $F$ be a set of at most $n-3$ faulty elements of $BP_n$ and $\{k_1,k_2,\ldots,k_m\}\subseteq \langle n\rangle$ such that $n\geq 4$ and $m\geq 6$.
If each of $m$ subgraphs $BP_n^{k_1}-F$, $BP_n^{k_2}-F$,
$\ldots$, $BP_n^{k_m}-F$ is Hamiltonian-connected, then  $\bigcup_{i=1}^mBP_i^{k_i} -F$ is Hamiltonian-connected as well.
\end{lemma}
\begin{proof}
Let $u,v$ be distinct vertices of the graph $\bigcup_{i=1}^mBP_i^{k_i} -F$. If they fall into distinct subgraphs, then a Hamiltonian path between them exists by Lemma~\ref{BPn1}, so we can assume that they fall into the same subgraph, say $BP_n^{k_1}$. By our assumption there is a Hamiltonian path $H_1[u,v]$ of $BP_n^{k_1}$. By part \eqref{prop:edge-sel:1} of Proposition~\ref{prop:edge-sel} there is  an edge $ab\in E(H_1[u,v])$ such that $aa^n,bb^n,a^n,b^n$ are fault-free and $H_1[u,v]=H_2[u,a]+\langle a,b\rangle+H_3[b,v]$. As Proposition~\ref{prop:edge-sel}\,\eqref{prop:edge-sel:1} also guarantees that the out-neighbors $a^n,b^n$ of $a,b$ fall into distinct subgraphs and we have $m-1\ge5$, by Lemma~\ref{BPn1} there is a~Hamiltonian path $H_4[a^n,b^n]$ of $\bigcup_{i=2}^mBP_i^{k_i} -F$. Then
$H[u,v]:=H_2[u,a]+\langle a,a^n\rangle+H_4[a^n,b^n]+\langle b^n,b\rangle+H_3[b,v]$ forms the desired Hamiltonian path of $\bigcup_{i=1}^mBP_i^{k_i} -F$.
\end{proof}
Now we are ready to extend Lemma~\ref{BPn1} from Hamiltonian paths to 2-disjoint path covers.
\begin{lemma}\label{2d}
Let $F$ be a set of at most $n-3$ faulty elements of $BP_n$ and $\{k_1,k_2,\ldots,k_m\}\subseteq \langle n\rangle$ such that $n\geq 4$ and $m\geq 5$.
If each of $m$ subgraphs $BP_n^{k_1}-F$, $BP_n^{k_2}-F$,
$\ldots$, $BP_n^{k_m}-F$ admits a~$2$-DPC joining arbitrary terminal sets of size two, then 
\begin{enumerate}[\upshape(1)]
  \item\label{2d:1}
 $\bigcup_{i=1}^mBP_n^{k_i}-F$ is Hamiltonian-connected,
\item \label{2d:2}
there is a 2-DPC of
$\bigcup_{i=1}^mBP_n^{k_i}-F$ formed by paths $P[u,v]$ and $Q[x,y]$ for arbitrary terminal vertices $u,v,x,y$ such that neither $\{u,v\}$ nor $\{x,y\}$ belong to $BP_n^{k_i}-F$ for some $i\in[m ]$,
\item\label{2d:3}
 if $m\ge6$, $\bigcup_{i=1}^mBP_n^{k_i}-F$ admits a 2-DPC for arbitrary terminal sets of size two.
\end{enumerate}
\end{lemma}
\begin{proof}
First deal with part \eqref{2d:2}.
As $m\geq 5$, by Proposition~\ref{prop} we can assume that $u\in V(BP_n^{k_1}-F)$, $v\in V(BP_n^{k_m}-F)$ and $k_i\neq \overline{k}_{i+1}$ for $i\in[m-1]$. Similarly, by Proposition~\ref{prop} we can rearrange $k_1,k_2,\cdots, k_m$ to obtain
$k_1',k_2',\cdots, k_m'$ such that $x\in V(BP_n^{k'_1}-F)$, $y\in V(BP_n^{k'_m}-F)$ and $k_i'\neq \overline{k}_{i+1}'$ for $i\in[m-1]$.

By Lemma~\ref{BPn}, for every pair $i,j\in\langle n\rangle$ such that $i\not\in\{j,\bar j\}$ we have  $|E_{i,j}(BP_n)|-|F|\geq(n-2)!\cdot 2^{n-2}-(n-3)>4$ for $n\geq 4$ and therefore $E_{i,j}(BP_n)$ contains at least four fault-free edges incident to fault-free vertices  in this case. Recall that $E_{i,j}(BP_n)$ forms a matching and therefore
these edges are incident with distinct vertices. It follows that for every $i\in[m-1]$ we can select edges
 $v_iu_{i+1}\in E_{k_i,k_{i+1}}(BP_n-F)$ and $y_ix_{i+1}\in E_{k'_i,k'_{i+1}}(BP_n-F)$ such that
\begin{itemize}
\item $v_i\in V(BP_n^{k_i}-F)$, $u_{i+1}\in V(BP_n^{k_{i+1}}-F)$, $y_i\in V(BP_n^{k'_i}-F)$ and $x_{i+1}\in V(BP_n^{k'_{i+1}}-F)$,
\item all the vertices of $\{u_i,v_i,x_i,y_i\}_{i=1}^m$ are pairwise distinct.
\end{itemize}
provided that we put $u_1:=u,$ $v_m:=v$, $x_1:=x$ and $y_m:=y$.

By our assumption, for each $i\in[m]$ there is a 2-DPC
of $BP_n^{k_i}-F$ formed by $\{P_i[u_i,v_i],Q_j[x_j,y_j]\}$ where $j\in[m]$ is such that $k'_j=k_i$.
Then
\begin{align*}
&P[u,v] := P_1[u,v_1] +  \langle v_1,u_2\rangle + \sum_{i=2}^{m-1}(P_i[u_i,v_i]+\langle v_{i},u_{i+1}\rangle) +  P_m[u_{m},v] \text{ and}\\
&Q[x,y] := Q_1[x,y_1] + \langle y_1,x_2\rangle + \sum_{i=2}^{m-1}(Q_i[x_i,y_{i}] + \langle y_{i},x_{i+1}\rangle) + Q_m[x_{m},y]
\end{align*}
form the desired 2-DPC of $\bigcup_{i=1}^mBP_n^{k_i}-F$ joining the given terminal sets, see Fig.~\ref{f2} for an illustration.

Next we show that part \eqref{2d:1} is an immediate corollary of part \eqref{2d:2}. 
To that end, let $u,v$ be distinct vertices of $\bigcup_{i=1}^mBP_i^{k_i}-F$. 
Select an arbitrary fault-free edge $ab$ in a subgraph, different from those containing $u$ and $v$, and apply part \eqref{2d:2} to obtain a 2-DPC of $\bigcup_{i=1}^mBP_i^{k_i}-F$ formed by the paths $P[u,a]$ and $Q[b,v]$. Then
\[
H[u,v]:=P[u,a]+\langle a,b\rangle+Q[b,v]
\]
is the desired Hamiltonian path of $\bigcup_{i=1}^mBP_i^{k_i}-F$.

Finally, to settle part \eqref{2d:3}, assume that $m\ge6$. If one of the subgraphs, say $BP_n^{k_1}-F$, contains exactly three terminals, say $u,v,x$, apply part~\eqref{prop:edge-sel:3} of Proposition~\ref{prop:edge-sel} to select a vertex $a$ of $BP_n^{k_1}-F-\{u,v,x\}$ such that both $a^n$ and $aa^n$ are fault-free while $a^n$ and $y$ fall into distinct subgraphs. By our assumption there is a 2-DPC $\{P[u,v],P_1[x,a]\}$ of $BP_n^{k_1}-F$, and by Lemma~\ref{BPn1} there is a Hamiltonian path $Q_1[a^n,y]$ of $\bigcup_{i=2}^mBP_n^{k_i}$. The desired 2-DPC of  $\bigcup_{i=1}^mBP_n^{k_i}-F$ is formed by $P[u,v]$ and $Q[x,y]:=P_1[x,a]+\langle a,a^n\rangle+Q_1[a^n,y]$.

If $u,v$ are contained in the same subgraph, say $BP_n^{k_1}-F$, by our assumption this subgraph contains  a Hamiltonian path $P[u,v]$. Next observe that by part \eqref{2d:1} of the current lemma (that we have just proved above), $\bigcup_{i=2}^mBP_n^{k_i}-F$ is Hamiltonian connected and therefore contains 
a~Hamiltonian path $Q[x,y]$. 
Then $P[u,v]$ and $Q[x,y]$ form the desired 2-DPC of $\bigcup_{i=1}^mBP_n^{k_i}-F$.

Finally, if all terminals $u,v,x,y$ are contained in the same subgraph, say $BP_n^{k_1}-F$, by our assumption this subgraph contains a~2-DPC formed by the paths $P[u,v]$ and $Q_1[x,y]$. By part~\eqref{prop:edge-sel:1} of Proposition~\ref{prop:edge-sel}, one of these paths, say $Q_1[x,y]$, contains an edge $ab$ such that both out-vertices $a^n,b^n$ and out-edges $aa^n,bb^n$ are fault-free. Moreover, $a^n$ and $b^n$ fall into distinct subgraphs. Note that the removal of $ab$ splits $Q_1[x,y]$ into two paths $Q_2[x,a]$ and $Q_3[b,y]$. Furthermore, Lemma~\ref{BPn1} guarantees the existence of a~Hamiltonian path $H[a^n,b^n]$ of $\bigcup_{i=2}^mBP_n^{k_i}-F$. The desired 2-DPC of $\bigcup_{i=1}^mBP_n^{k_i}-F$ is then formed by $P[u,v]$ and
\[
Q[x,y]:=Q_2[x,a]+\langle a,a^n\rangle+H[a^n,b^n]+\langle b^n,b\rangle+Q_3[b,y] .
\]
\end{proof}
\begin{figure}[htb]
\begin{center}
\includegraphics[width=\textwidth]{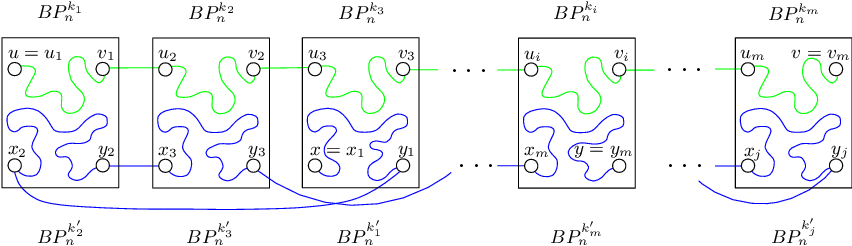}
\caption{Illustration of the proof of part~\eqref{2d:2} of Lemma~\ref{2d}. The paths $P[u,v]$ and $Q[x,y]$ are denoted by green and blue colors, respectively. The indices $i,j\in[m]$ are selected so that $k_i=k'_{m}$ and $k'_j=k_m$. Please note that in general it may not be true that $(k_1,k_2,k_3)=(k'_2,k'_3,k'_1)$ or $k_m\ne k'_{m}$ as depicted in this illustration.}
\label{f2}
\end{center}
\end{figure}
As a corollary we obtain another extension of Lemma~\ref{BPn1}. This time we show that Hamiltonian-connectivity is preserved even if one of the subgraphs is only Hamiltonian provided that the others admit 2-DPC's for arbitrary terminal sets and the number of subgraphs is sufficiently large.
\begin{lemma}\label{lem:Ham-con-one-cycle}
Let $n\geq 4$, $\{k_1,k_2,\dots,k_m\}\subseteq\langle n\rangle$ such that  $m=|\langle n \rangle|-1$ and $F$ be a set of at most $n-3$ faulty elements of $BP_n^{k_j}$ for some $j\in[m]$. Suppose that $BP_n^{k_j}-F$ is Hamiltonian while  $BP_n^{k_i}$ for each $i\in[m]\setminus\{j\}$ admits a 2-DPC for arbitrary terminal sets of size two. Then there is a~Hamiltonian path  in $\bigcup_{i=1}^mBP_n^{k_i}-F$ between $u$ and $v$ provided that
\begin{enumerate}[\upshape(1)]
\item\label{lem:Ham-con-one-cycle:1}
 either  $u\in V(BP_n^{k_1})$ and $v\in V(BP_n^{k_m})$,
\item\label{lem:Ham-con-one-cycle:2} 
or $u,v\in V(BP_n^{k_1})$ and $j>1$.
\end{enumerate}
\end{lemma}
\begin{proof}
Let $C$ denote the Hamiltonian cycle of $BP_n^{k_j}-F$. First assume that $u$ and $v$ are chosen so that (1) holds.
We start with the case that $j=1$ or $j=m$, without loss of generality assuming the former, which means that $u$ lies on the cycle $C$ in $BP_n^{k_1}-F$. By Lemma~\ref{BPn0}, the two neighbors of $u$ on $C$ have out-edges leading to different subgraphs, and therefore we can select a neighbor $a$ of $u$ on $C$ such that its out-neighbor $a^n$ falls into a subgraph $BP_n^{k_l}$, $2\le l\le m$.
Observe that the vertices of $C$ form a Hamiltonian path $H_1[u,a]$ of $BP_n^{k_1}-F$. Recall that our assumptions imply that each subgraph $BP_n^{k_i}$, $2\le i\le m$ is Hamiltonian-connected.
If $a^n\ne v$, Lemma~\ref{lem:ham-conn} guarantees the existence of a Hamiltonian path $H_2[a^n,v]$ of $\bigcup_{i=2}^m BP_n^{k_i}$. Then  $H[u,v]:=H_1[u,a]+\langle a,a^n\rangle+H_2[a^n,v]$ is the desired Hamiltonian path of $\bigcup_{i=1}^m BP_n^{k_i}-F$.

If, however, $a^n=v$, by part \eqref{prop:edge-sel:2} of Proposition~\ref{prop:edge-sel} we can select an edge $bc$ of $C$ such that $\{b,c\}\cap\{u,a\}=\emptyset$ while both out-neighbors of $b$ and $c$ fall into $\bigcup_{i=2}^m BP_n^{k_i}$. By removing edges $bc$ and $ua$, we split $C$ into paths $P[u,b]$ and $Q[c,a]$ (or $P[u,c]$ and $Q[b,a]$, but we can without loss of generality assume the former), thus forming a~2-DPC of $BP_n^{k_1}-F$. 
Note that our assumptions on $b$ and $c$ imply that $v\not\in\{b^n,c^n\}$, and therefore Lemma~\ref{lem:ham-conn} guarantees the existence of a Hamiltonian path $H_2[b^n,c^n]$ of $\bigcup_{i=2}^m BP_n^{k_i}-v$. Then the desired Hamiltonian path of $\bigcup_{i=1}^m BP_n^{k_i}-F$ is formed by
\[
H[u,v]:=P[u,b]+\langle b,b^n\rangle+H_2[b^n,c^n]+\langle c^n,c\rangle+Q[c,a]+\langle a,a^n=v\rangle .
\]

It remains to settle the case that neither $u$ nor $v$ belongs into $BP_n^{k_j}$. Apply part~\eqref{prop:edge-sel:2} of Proposition~\ref{prop:edge-sel}
to select an edge $ab$ on the cycle $C$ such that the out-neighbors $a^n,b^n$ belong to $\bigcup_{i\in[m]\setminus\{j\}}BP_n^{k_i}$ while $\{a^n,b^n\}\cap\{u,v\}=\emptyset$. Hence by Lemma~\ref{2d} there is a~2-DPC of $\bigcup_{i\in [m]\setminus\{j\}}BP_n^{k_i}$ formed by the paths $P[u,a^n]$ and $Q[b^n,v]$. Let $H_3[a,b]$ be the Hamiltonian path of $BP_n^{k_j}-F$ formed by the vertices of the cycle C. Then
\[
H[u,v]:=P[u,a^n]+\langle a^n,a\rangle+H_3[a,b]+\langle b,b^n\rangle + Q[b^n,v]
\]
is the desired Hamiltonian path of $\bigcup_{i=1}^m BP_n^{k_i}-F$.
\end{proof}
The next lemma settles a special case needed for the construction in the proof of our main result.
\begin{lemma}\label{lem:spec-case-2-2}
Let $n\geq 4$, $k_1,k_2\in\langle n\rangle$ such $k_1\ne k_2$ and $F$ be a set of at most $n-3$ faulty elements of $BP_n^{k_1}$. Suppose that $BP_n^{k_1}-F$ is Hamiltonian-connected while  $BP_n^{k}$ for each $k\in\langle n\rangle\setminus\{k_1\}$ admits a~2-DPC joining arbitrary terminal sets of size two. Then $BP_n-F$ admits a~2-DPC joining arbitrary terminal sets of size two provided that each of $BP_n^{k_1}-F$ and $BP_n^{k_2}$ contains exactly two terminals. 
\end{lemma}
\begin{proof}
Let $x,y,u,v$ be pairwise distinct vertices such that $x,y\in V(BP_n^{k_1}-F)$ and $u,v\in V(BP_n^{k_2})$. Our goal is to construct 2-DPC's of $BP_n-F$ formed by $\{P[x,y],Q[u,v]\}$ as well as by $\{P[x,u], Q[y,v]\}$.   

As $BP_n^{k_1}-F$ is Hamiltonian-connected, it contains a Hamiltonian path $H[x,y]$. By part~\eqref{prop:edge-sel:1} of Proposition~\ref{prop:edge-sel} this path contains an edge $ab$ such that $\{a,b,a^n,b^n\}\cap\{x,y,u,v\}=\emptyset$. The removal of the edge $ab$ splits $H[x,y]$ into two paths, say $P_1[x,a]$ and $Q_1[y,b]$, thus forming a~2-DPC of $BP_n^{k_1}-F$.

Furthermore, as $|\langle n \rangle\setminus\{k_1\}|=2n-1>6$ for $n\ge4$, Lemma~\ref{2d} implies that $\bigcup_{k\in\langle n \rangle\setminus\{k_1\}}BP_n^{k}$ admits a~2-DPC formed by  $\{P_2[a^n,b^n],Q_2[u,v]\}$ as well as by $\{P_3[a^n,u],Q_3[b^n,v]\}$. The first desired 2-DPC of $BP_n-F$ is then formed by 
\begin{align*}
&P[x,y]:=P_1[x,a]+\langle a,a^n\rangle+P_2[a^n,b^n]+\langle b^n,b\rangle+Q_1[b,y],\\
&Q[u,v]:=Q_2[u,v],
\end{align*}
while the other 2-DPC is defined by
 \begin{align*}
&P[x,u]:=P_1[x,a]+\langle a,a^n\rangle+P_3[a^n,u],\\
&Q[y,v]:=Q_1[y,b]+\langle b,b^n\rangle+Q_3[b^n,v].
\end{align*}	
\end{proof}
The last result of this section establishes the basis needed for the inductive construction in the proof of our main result.
\begin{theorem}\label{bp3}
$BP_3$ has a 2-DPC joining arbitrary terminal sets of size two.
\end{theorem}
\begin{proof}
Since $BP_3$ is a vertex transitive graph, we can fix one terminal vertex to $123\in V(BP_3^3)$. By a~computer search we were able to find the 
2-DPC's of $BP_3$ for all different distributions of the remaining three terminals. The list of the resulting paths is available online at Mendeley Data \cite{DvorakGu23}.
\end{proof}

\section{Main result}
\label{sec:main}

At this point, we are well armed to formulate and prove the main result of this paper.

\begin{theorem}\label{main}
Let $n\ge4$ and $F$ be a set of at most $n-4$ faulty elements in $BP_n$. Then there is a~2-DPC of
$BP_n-F$ joining arbitrary terminal sets of size two.
\end{theorem}

\begin{proof}
We argue by induction on $n$. As the case $n=4$ follows from part~\eqref{2d:3} of Lemma~\ref{2d} using Theorem~\ref{bp3}, we can assume that $n\ge5$ and that the statement of the theorem holds for $n-1$. Recall that $BP_n$ can be partitioned into $2n\ge8$ vertex-disjoint subgraphs $BP_n^i$, $i\in\langle n\rangle$,  where each subgraph is isomorphic to $BP_{n-1}$. First observe that if $BP_n^{i}$ contains no more that $n-5$ faulty elements for all $i\in \langle n\rangle$,
then, by the induction hypothesis, each of $m$ subgraphs $BP_n^{k_1}-F$, $BP_n^{k_2}-F$,
$\ldots$, $BP_n^{k_m}-F$ admits a~$2$-DPC joining arbitrary terminal sets of size two. Then the statement of the theorem holds by Lemma~\ref{2d}.

It remains to settle the case that there is a $k^\ast\in \langle n\rangle$ such that $BP_n^{k^\ast}$ contains exactly $n-4$ faulty elements while the remaining subgraphs as well as all the out-edges are fault-free. To that end, let $u,v,x,y$ be distinct vertices of $BP_n-F$. Our goal is to construct vertex disjoint paths $P[u,v]$ and $Q[x,y]$ to form a~2-DPC of $BP_n-F$. 
Recall that the induction hypothesis implies that for each $k\in \langle n\rangle\setminus\{k^\ast\}$,  the (fault-free) subgraph $BP_n^{k}$  admits a~$2$-DPC joining arbitrary terminal sets of size two. Note that this implies that each of these subgraphs is also Hamiltonian-connected.  
Moreover, the subgraph $BP_n^{k^\ast}$ is Hamiltonian-connected by Lemma~\ref{BPn2}. 

We consider four cases according to the distribution of the terminal vertices.

\textsc{Case 1}. All four terminal vertices are contained in the same subgraph $BP_n^{k_1}-F$ for some $k_1\in \langle n\rangle$.

\textsc{Subcase 1.1}. $k^\ast=k_1$.

Choose an arbitrary faulty element $e\in F$ and set $F'=F\setminus\{e\}$. As now $|F'|=(n-1)-4$, by the induction hypothesis there is a~2-DPC
$\{P_1[u,v],Q_1[x,y]\}$ of $BP_n^{k_1}-F'$.
If $e$ does not lie on $P_1[u,v]$ or $Q_1[x,y]$, select an arbitrary edge $ab$ on $Q_1[x,y]$. If $e$ lies on $P_1[u,v]$ or $Q_1[x,y]$, we can without loss of generality assume the latter. If $e$ is a faulty edge, let $a,b$ be the vertices of $Q_1[x,y]$ incident with $e$. In both cases we have
\[
Q_1[x,y]=Q_2[x,a]+\langle a,b\rangle+Q_3[b,y].
\]
If $e$ is a~faulty vertex $w$ visited by $Q_1[x,y]$, let $a,b$ be its neighbors on this path so that
\[
Q_1[x,y]=Q_2[x,a]+\langle a,w,b\rangle+Q_3[b,y].
\]
Recall that we have all faulty elements inside $BP_n^{k_1}$ and therefore the out-edges and out-neighbors of $a$ and $b$ are also fault-free.
In any case we have $d(a,b)\le2$ and therefore, by Lemma~\ref{BPn0}, the out-neighbors $a^n$ and $b^n$ of $a$ and $b$ belong to different subgraphs $BP_n^{k_2}$ and $BP_n^{k_3}$, respectively.
It remains to recall that $BP_n^i$ is Hamiltonian-connected for each $i\in\langle n\rangle\setminus\{k_1\}$ and therefore, by Lemma~\ref{BPn1}, $BP_n-BP_n^{k_1}$ contains a Hamiltonian path  $H[a^n,b^n]$ between $a^n$ and $b^n$.
Then
\begin{align*}
&P[u,v]:=P_1[u,v] \text{ and}\\
&Q[x,y]:=Q_2[x,a]+\langle a,a^n\rangle + H[a^n,b^n]+ \langle b^n,b\rangle+Q_3[b,y]
\end{align*}
form the desired 2-DPC of $BP_n-F$.

\textsc{Subcase 1.2}. $k^\ast=k_2\in\langle n\rangle\setminus\{k_1\}$.

As $BP_n^{k_1}$ contains no faulty elements in this case, by the induction hypothesis it admits a~2-DPC $\{P[u,v], Q_1[x,y]\}$. By Proposition~\ref{prop:edge-sel} one of the paths, say $Q_1[x,y]$, contains an edge $ab$ such that $a^n\in V(BP_n^{k_2}-F)$ and $b^n\in V(BP_n^{k_3})$ for some $k_3\in\langle n \rangle\setminus\{k_1,k_2\}$. We can without loss of generality assume that $a$ is closer to $x$ than $b$ on $Q_1[x,y]$, exchanging the roles of $x$ and $y$ is necessary, which means that $Q_1[x,y]=Q_2[x,a]+\langle a,b\rangle+Q_3[b,y]$. By part~\eqref{lem:Ham-con-one-cycle:1} of Lemma~\ref{lem:Ham-con-one-cycle} there is a~Hamiltonian path $H[a^n,b^n]$ of $BP_n-F-BP_n^{k_1}$.
Put
 \[
Q[x,y]:=Q_2[x,a]+\langle a,a^n\rangle + H[a^n,b^n]+\langle b^n,b \rangle+ Q_3[b,y]
\]
and observe that then $\{P[u,v],Q[x,y]\}$ forms the desired 2-DPC of $BP_n-F$.


\textsc{Case 2}. Terminal vertices are contained in two distinct subgraphs $BP_n^{k_1}$ and $BP_n^{k_2}$, where $k_1,k_2\in \langle n\rangle$ and $k_1\neq k_2$. We distinguish three subcases.

\textsc{Subcase 2.1}. One subgraph contains exactly three terminals. We can without loss of generality assume that $u,v,x$ fall into $ BP_n^{k_1}-F$ while $y$ lies in $BP_n^{k_2}-F$.

(2.1.1) $k^\ast\ne k_1$.

By part~\eqref{prop:edge-sel:3} of  Proposition~\ref{prop:edge-sel} we can select a vertex $w$ in $BP_n^{k_1}-F$ such that $w\not\in\{u,v,x\}$, its out-neighbor $w^n$ falls into $BP_n^{k_3}-F$ for some $k_3\in\langle n\rangle\setminus\{k_1,k_2\}$ and the out-edge $ww^n$ is fault-free. 
Recall that $y$ lies in $BP_n^{k_2}-F$ and therefore $w^n\ne y$.
As $k_1\ne k^\ast$,
$BP_n^{k_1}$ admits a~2-DPC $\{P_1[u,v],Q_1[x,w]\}$.
Recall that each of the remaining subgraphs is Hamiltonian-connected and therefore $BP_n-F-BP_n^{k_1}$ contains a Hamiltonian path $H[w^n,y]$ by Lemma~\ref{BPn1}. Then
\begin{align*}
&P[u,v]:=P_1[u,v] \text{ and}\\
&Q[x,y]:=Q_1[x,w]+\langle w,w^n\rangle + H[w^n,y]
\end{align*}
form the desired 2-DPC of $BP_n-F$. 

(2.1.2) $k^\ast=k_1$. 

Here, we provide two different constructions depending on whether the terminals $x$ and $y$ are adjacent or not.

(2.1.2.1) $x^n=y$.

As $BP_n^{k_1}-F$ is Hamiltonian-connected, it contains a Hamiltonian path $H_1[u,v]$. Recall that $x$ lies in $BP_n^{k_1}-F-\{u,v\}$, which means that $H_1[u,v]$ passes through $x$ and there are two distinct neighbors, say $a$ and $b$, of $x$ on this path. Then the removal of $x$ splits $H_1[u,v]$ into two paths $P_1[u,a]$ and $Q_1[b,v]$, thus forming a 2-DPC of $BP_n^{k_1}-F-x$. Note that it may happen that $a=u$ or $b=v$, but none of these equalities jeopardizes the construction that follows.
Recall that the out-edges leading out of $BP_n^{k_1}$ form a matching and therefore we have $y=x^n\not\in\{a^n,b^n\}$. 
As $(n-1)-3\ge|\{y\}|$, $BP_n^{k_2}-y$ is Hamiltonian-connected by Lemma~\ref{BPn2}. Consequently, $BP_n-BP_n^{k_1}-y$ is Hamiltonian-connected by Lemma~\ref{BPn1} and therefore it contains a~Hamiltonian path $H_2[a^n,b^n]$. 
Then
\begin{align*}
&P[u,v]:=P_1[u,a]+\langle a,a^n\rangle+H_2[a^n,b^n]+\langle b^n,b\rangle+ Q_1[b,v]\\
&Q[x,y]:=\langle x,y\rangle
\end{align*}
form the desired 2-DPC of $BP_n-F$, see Fig.~\ref{fig:3a}.

(2.1.2.2) $x^n\ne y$.

As $|F\cup\{x\}|=(n-1)-2$, by Lemma~\ref{BPn2} there is a Hamiltonian cycle $C$ of $BP_n^{k_1}-F-x$.
Let $a$ and $b$ be discordant neighbors of $u$ and $v$ on $C$. Note that 
\[
\ell(C)=|V(BP_n^{k_1}-F-x)|=(n-1)!\cdot 2^{n-1}-(n-3)>4
\]
for $n\ge5$ and therefore
we can select $a,b$ in such a way that $a\ne b$. It follows that the removal of edges $ua$ and $xb$ splits $C$ into two paths $P_1[u,v]$ and $Q_1[a,b]$ forming a 2-DPC of $BP_n^{k_1}-F-x$.

Furthermore, as the out-edges leading out of  $BP_n^{k_1}$ form a matching, it must be the case that the vertices $a^n,b^n$ and $x^n$ are pairwise distinct.
If $y\not\in\{a^n,b^n\}$, then, by the induction hypothesis, $BP_n-BP_n^{k_1}$ satisfies the assumptions of Lemma~\ref{2d} and therefore admits a~2-DPC $P_2[x^n,a^n],Q_2[b^n,y]$. On the other hand, if 
$y=a^n$ or $y=b^n$, then 
we can without loss of generality assume the latter and put $Q_2[b^n,y]=\langle y\rangle$. Further, as $BP_n^{k_2}-y$ is Hamiltonian connected by Lemma~\ref{BPn2}, it follows that $BP_n-BP_n^{k_1}-y$ is Hamiltonian-connected by Lemma~\ref{BPn1} and therefore contains a Hamiltonian path $P_2[x^n,a^n]$. In both cases, set
\begin{align*}
&P[u,v]:=P_1[u,v] ,\\
&Q[x,y]:=\langle x,x^n\rangle+
P_2[x^n,a^n]+\langle a^n,a\rangle + Q_1[a,b]+ \langle b,b^n\rangle+Q_2[b^n,y]
\end{align*}
and observe that then $\{P[u,v],Q[x,y]\}$ forms the desired 2-DPC of $BP_n-F$, see Fig.~\ref{fig:3b}.

\begin{figure}[htb]
\centering 
\subfloat[$x^n=y$\label{fig:3a}]
{\includegraphics[width=0.35\textwidth]{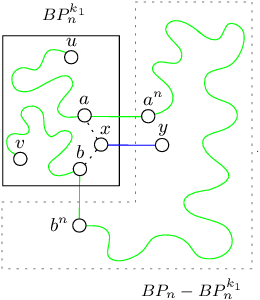} }
\qquad\qquad
\subfloat[$x^n\ne y$\label{fig:3b}]
{\includegraphics[width=0.35\textwidth]{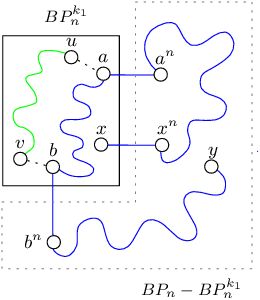} }
\caption{Construction of a 2-DPC in the subcase (2.1.2). The paths $P[u,v]$ and $Q[x,y]$ are depicted by green and blue colors, respectively. The subgraph $BP_n-BP_n^{k_1}$ is enclosed by a dotted line.}
\label{f3}
\end{figure}

\textsc{Subcase 2.2}. $u,v\in V(BP_n^{k_1}-F)$, $x,y\in V(BP_n^{k_2}-F)$ for some $k_1,k_2\in\langle n \rangle$, $k_1\ne k_2$.

Recall that each $BP_n^{k}$, $k\in\langle n \rangle$, is Hamiltonian-connected and therefore there is a~Hamiltonian  path $P[u,v]$ of $BP_n^{k_1}-F$. Moreover, by Lemma~\ref{lem:ham-conn} there is a Hamiltonian path $Q[x,y]$ of $BP_n-F-BP_n^{k_1}$ which together with $P[u,v]$ forms the desired 2-DPC of $BP_n-F$.

\textsc{Subcase 2.3}. $u,x\in V(BP_n^{k_1}-F)$, $v,y\in V(BP_n^{k_2}-F)$ for some $k_1,k_2\in\langle n \rangle$, $k_1\ne k_2$.

If $k^\ast\in\{k_1,k_2\}$, then the desired 2-DPC of $BP_n-F$ exists by Lemma~\ref{lem:spec-case-2-2}. We can therefore assume that 
$k^\ast\in \langle n\rangle\setminus\{k_1,k_2\}$. Note that we can without loss of generality assume that $k^\ast\neq \overline{k_1}$, swapping the roles of $k_1$ and $k_2$ if necessary.
 
By part~\eqref{prop:edge-sel:1} of Proposition~\ref{prop:edge-sel} there is a~fault-free vertex  $a$ of $BP_n^{k^\ast}$
such that $a^n$ belongs to $BP_n^{k_1}-u-x$. 
By Lemma~\ref{BPn2}, there exists a Hamiltonian cycle $C$ of $BP_n^{k^\ast}-F$. Then $a$ has two neighbors on $C$ and by Lemma~\ref{BPn0} one of them has its out-neighbor in a subgraph different from both $BP_n^{k_1}$ and $BP_n^{k_2}$. So let $b$ be the neighbor of $a$ on $C$ 
such that $b^n\in BP_n^{k_3}$, where $k_3\in \langle n\rangle\setminus \{k_1,k_2,{k^\ast}\}$. Note that the vertices of $C$ form a~Hamiltonian path $H[a,b]$ of $BP_n^{k^\ast}-F$. 
As $|\langle n \rangle|\ge10$, we can select a 
$k_4\in \langle n\rangle\setminus \{k_1,k_2,k_3,{k^\ast},\overline{k_1}\}$.
Using part~\eqref{prop:edge-sel:1} of  Proposition~\ref{prop:edge-sel} again, we can select a vertex $c$  in $BP_n^{k_1}-\{u,x,a^n\}$ such that $c^n\in BP_n^{k_4}$. 
By the induction hypothesis, there is a~2\nobreakdash-DPC of $BP_n^{k_1}$ consisting of the paths $P_1[u,a^n]$ and $Q_1[x,c]$. 
By Lemma~\ref{2d}, there is a~2\nobreakdash-DPC of $BP_n-BP_n^{k_1}-BP_n^{k^\ast}$ formed by the paths $P_2[b^n,v]$ and $Q_2[c^n,y]$. The desired 2-DPC of $BP_n-F$ is then formed by 
\begin{align*}
&P[u,v]:=P_1[u,a^n]+\langle a^n,a\rangle+H[a,b]+\langle b,b^n\rangle+P_2[b^n,v],\\
&Q[x,y]:=Q_1[x,c]+\langle c,c^n\rangle+Q_2[c^n,y],
\end{align*}
see Fig.~\ref{f4}.

\begin{figure}[htb]
\begin{center}
\includegraphics[width=0.6\textwidth]{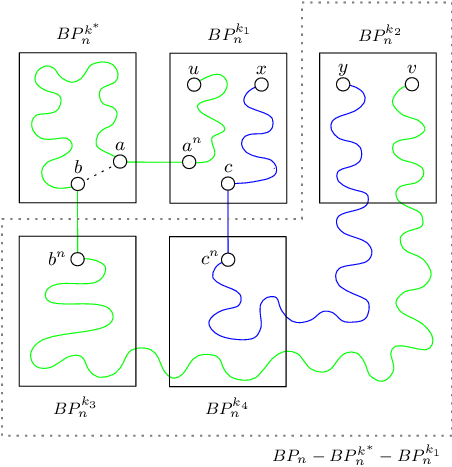}
\caption{Construction of a 2-DPC in the \textsc{Subcase 2.3}. The paths $P[u,v]$ and $Q[x,y]$ are depicted by green and blue colors, respectively. The subgraph $BP_n-BP_n^{k^\ast}-BP_n^{k_1}$ is enclosed by a dotted line.}
\label{f4}
\end{center}
\end{figure}

\textsc{Case 3}. Terminal vertices are contained in three distinct subgraphs $BP_n^{k_1}-F$, $BP_n^{k_2}-F$ and $BP_n^{k_3}-F$, where $k_1,k_2,k_3\in \langle n\rangle$ are pairwise distinct.

\textsc{Subcase 3.1}. $u,v\in V(BP_n^{k_1}-F)$, $x\in V(BP_n^{k_2}-F)$, $y\in V(BP_n^{k_3}-F)$. 

Recall that the subgraph $BP_n^k-F$
\begin{itemize}
\item is Hamiltonian-connected for each $k\in\langle n\rangle$ by Lemma~\ref{BPn2},
\item admits a~$2$-DPC joining arbitrary terminal sets of size two for each $k\in\langle n\rangle\setminus\{k^\ast\}$ by the induction hypothesis.
\end{itemize} 

The desired 2-DPC of $BP_n-F$ is formed by a~Hamiltonian path $P[u,v]$ of $BP_n^{k_1}-F$ and a~Hamiltonian path $Q[x,y]$ of $BP_n-F-BP_n^{k_1}$ which exists by part~\eqref{lem:Ham-con-one-cycle:1} of Lemma~\ref{lem:Ham-con-one-cycle}.

\textsc{Subcase 3.2}. $u,x\in V(BP_n^{k_1}-F)$, $v\in V(BP_n^{k_2}-F)$, $y\in V(BP_n^{k_3}-F)$. 

(3.2.1) $k^\ast\in \langle n\rangle\setminus\{k_1,k_2,k_3\}$. 

Note that we have ${k^\ast}\neq \overline{k_2}$ or ${k^\ast}\neq \overline{k_3}$ and we can without loss of generality assume the latter. 
As $BP_n^{k^\ast}-F$ is Hamiltonian-connected, 
by part \eqref{prop:edge-sel:1} of Proposition~\ref{prop:edge-sel} it contains 
vertices $a,b$ such that $a^n$ and $b^n$ fall into $BP_n^{k_3}-y$ and $BP_n-BP_n^{k^\ast}-BP_n^{k_3}-\{u,v,x\}$, respectively. 
Let $H_1[a,b]$ be a~Hamiltonian path  of $BP_n^{k^\ast}-F$.
By Lemma~\ref{2d} there are paths $P[u,v]$ and $Q_1[x,b^n]$ forming a~2-DPC of $BP_n-F-BP_n^{k^\ast}-BP_n^{k_3}$.
As $BP_n^{k_3}$ is Hamiltonian-connected, it contains a~Hamiltonian path $H_2[a^n,y]$ of $BP_n^{k_3}$. 
The desired 2-DPC of $BP_n-F$ is then formed by
$P[u,v]$ and 
\[
Q[x,y]:=Q_1[x,b^n]+\langle b^n,b\rangle+H_1[b,a]+\langle a,a^n\rangle+H_2[a^n,y].
\]

(3.2.2) ${k^\ast}={k_1}$.

As $BP_n^{k_1}-F$ is Hamiltonian-connected, it contains a Hamiltonian path $P[u,x]$. By part \eqref{prop:edge-sel:1} of Proposition~\ref{prop:edge-sel}, the path contains an edge $ab$ such that $\{a,b\}\cap\{u,x\}=\emptyset=\{a^n,b^n\}\cap\{v,y\}$. Removal of the edge $ab$  splits $P[u,x]$ into two paths, say $P_1[u,a]$ and $Q_1[b,x]$, forming a 2-DPC of $BP_n^{k_1}-F$.
Recall that for each of $k\in \langle n\rangle\setminus\{k_1\}$,  the subgraph $BP_n^{k}$  admits a~$2$-DPC joining arbitrary terminal sets of size two. Consequently, by Lemma~\ref{2d} there is a 2-DPC of $BP_n-F-BP_n^{k_1}$ formed by $P_2[a^n,v]$ and $Q_2[b^n,y]$.
The desired 2-DPC of $BP_n-F$ then consists of
\begin{align*}
&P[u,v]:=P_1[u,a]+\langle a,a^n\rangle+P_2[a^n,v],\\
&Q[x,y]:=Q_1[x,b]+\langle b,b^n\rangle+Q_2[b^n,y].
\end{align*}
 
(3.2.3) ${k^\ast}={k_2}$ or ${k^\ast}={k_3}$.

Without loss of generality assuming the former, recall that $BP_n^{k_2}$ is Hamiltonian-connected. As $|\langle n\rangle|\ge10$, we can select a $k_4\in \langle n\rangle\setminus\{k_1,k_2,\overline{k_2},k_3\}$. Then, by part~\eqref{prop:edge-sel:1} of Proposition~\ref{prop:edge-sel}, there is a fault-free vertex $a$ of $BP_n^{k_2}$ such that $a^n\in BP_n^{k_4}$. Let $H[v,a]$ be a Hamiltonian path of $BP_n^{k_2}-F$.	
Since $a^n\notin\{u,x,y\}$, by Lemma~\ref{2d} we obtain a 2-DPC of $BP_n-BP_n^{k_2}$ consisting of the paths $P_1[a^n,u]$ and $Q[x,y]$. 
The desired 2-DPC of $BP_n-F$ is then formed by
$Q[x,y]$ and 
\[P[v,u]:=H[v,a]+\langle a,a^n\rangle+P_1[a^n,u]\,.\]

\textsc{Case 4}. Terminal vertices are contained in four distinct subgraphs: $u\in V(BP_n^{k_1}-F)$, $x\in V(BP_n^{k_2}-F)$, $v\in V(BP_n^{k_3}-F)$, $y\in V(BP_n^{k_4}-F)$, where $k_i\in \langle n\rangle$ for $i\in [4]$ are pairwise distinct.

\textsc{Subcase 4.1}. 
$k^\ast={k_i}$ for some $i\in [4]$. 
We can without loss of generality assume that $k^\ast={k_1}$. 

Recall that $BP_n^{k_1}-F$ is Hamiltonian-connected.  
By part~\eqref{prop:edge-sel:1} of Proposition~\ref{prop:edge-sel} we can select a fault-free vertex $a$ of $BP_n^{k_1}$ such that $a\ne u$ and $a^n\not\in\{v,x,y\}$.
Let $H[u,a]$ be a Hamiltonian path of $BP_n^{k_1}-F$. By Lemma~\ref{2d} there is a~2-DPC of $BP_n-BP_n^{k_1}$ consisting of the paths $P_1[a^n,v]$ and $Q[x,y]$. 
The desired 2-DPC of $BP_n-F$ is then formed by $Q[x,y]$ and 
\[P[u,v]:=H[u,a]+\langle a,a^n\rangle+P_1[a^n,v].\]

\textsc{Subcase 4.2}. $k^\ast\in \langle n\rangle\setminus \{k_1,k_2,k_3,k_4\}$.

Note that in this case we have ${k^\ast}\neq\overline{k_1}$ or ${k^\ast}\neq\overline{k_2}$. Without loss of generality assuming the former, recall that $BP_n^{k^\ast}$ is Hamiltonian-connected. Therefore,  by part~\eqref{prop:edge-sel:1} of Proposition~\ref{prop:edge-sel} we can select distinct vertices $a$ and $b$ in $BP_n^{k^\ast}-F$ such that $a^n$ falls into $BP_n^{k_1}-u$ while $b^n\not\in V(BP_n^{k_1})\cup\{v,x,y\}$.

As $BP_n^{k_1}$ is Hamiltonian-connected as well, there are Hamiltonian paths $H_1[u,a^n]$ and $H_2[a,b]$ of $BP_n^{k_1}$ and $BP_n^{k^\ast}-F$, respectively.
By Lemma~\ref{2d} there are paths $P_1[b^n,v]$ and $Q[x,y]$ forming a 2-DPC of $BP_n-BP_n^{k_1}-BP_n^{k^\ast}$.
The desired 2-DPC of $BP_n-F$ is then formed by $Q[x,y]$ and
\[P[u,v]:=H_1[u,a^n]+\langle a^n,a\rangle+H_2[a,b]+\langle b,b^n\rangle+P_1[b^n,v].\]

\end{proof}

\section{Concluding remarks}
Is the upper bound $n-4$ on the number of faulty elements in Theorem~\ref{main} sharp? It is easy to see that for every $n\ge3$ there exists a set $F$ of $n-2$ faulty edges or faulty vertices of $BP_n$ such that a 2-DPC of $BP_n-F$ joining certain terminal sets does not exist. 

Indeed, let $u$ and $x$ be two vertices  at distance two in $BP_n$ and $w$ be their common neighbor. Let $F_e$ be the set of $n-2$ faulty edges incident with $w$ but distinct from $wu$ and $wx$. Then $w$ is incident with exactly two fault-free edges and therefore any fault-free path passing through $w$ must visit both $u$ and $x$.
It follows that $BP_n-F_e$ does not admit a 2-DPC consisting of paths $P[u,v]$ and $Q[x,y]$ for any terminals $v$ and $y$. An analogical statement holds for $BP_n-F_v$ where $F_v$ is the set of all vertices incident with edges of $F_e$ except $w$. 
  
Consequently, the result of Theorem~\ref{bp3} is optimal, but there is still a small gap between the upper and lower bounds on the number of faults tolerated by Theorem~\ref{main}. Note that in Section~\ref{sec:tools}, which provides tools for the constructive proof of the main result, all the statements are valid also for the case of $n-3$ faulty elements. It is therefore natural to conclude this paper with an open problem: Is it possible to improve the upper bound on the number of faulty elements provided by Theorem~\ref{main} to $n-3$? 

\section*{Acknowledgments} 
This work was partially supported by the Czech Science Foundation Grant 19-08554S (Tom\'a\v{s} Dvo\v{r}\'ak) and the National Natural Science Foundation of China Grants  12101610 and 11971054 (Mei-Mei Gu). 


\begin{thebibliography}{99} \small
\baselineskip=12pt
\parskip 0pt

\bibitem{bbp}
 S. Blancoa, C. Buehrle and  A. Patidar,
Cycles in the burnt pancake graph, Discrete Appl. Math. 271 (2019), 1--14, \doi{10.1016/j.dam.2019.08.008}.

\bibitem{Chin-2009}
C. Chin, T.-H. Weng, L.-H. Hsu and S.-C. Chiou,
The spanning connectivity of the burnt pancake graphs,
IEICE Trans. Inform. Syst. E92-D (2009), 389--400, 
\doi{10.1587/transinf.E92.D.389}.

\bibitem{Compeau-2011}
P.~E.~C. Compeau,
Girth of pancake graphs,
Discrete Appl. Math. 159 (2011), 1641--1645, \doi{10.1016/j.dam.2011.06.013}.

\bibitem{Dvorak08}
T. Dvo\v{r}\'ak, P. Gregor, 
Partitions of faulty hypercubes into paths with prescribed endvertices, SIAM J. Discrete Math. 22 (2008), 
1448--1461, \doi{10.1137/060678476}.

\bibitem{Dvorak17}
T. Dvo\v{r}\'ak, P. Gregor, V. Koubek, 
Generalized Gray codes with prescribed ends, Theor. Comput. Sci. 668 (2017), 
70--94, \doi{10.1016/j.tcs.2017.01.010}.

\bibitem{DvorakGu23}
T. Dvo\v{r}\'ak, M.-M. Gu, Paired 2-disjoint path covers of the 3-dimensional burnt pancake graph, Mendeley Data, 2023, \doi{10.17632/dg256vjs6x.1}.

\bibitem{Enomoto00}
H. Enomoto, K. Ota, 
Partitions of a~graph into paths with prescribed endvertices and lengths, 
J. Graph Theory 34 (2000), 
163--169, \newline
\doi{10.1002/1097-0118(200006)34:2<163::AID-JGT5>3.0.CO;2-K}.

\bibitem{Gates-1979}
W.H. Gates and C.H. Papadimitriou,
Bounds for sorting by prefix reversal,
Discrete Math. 27 (1979), 47--49, \doi{10.1016/0012-365X(79)90068-2}.

\bibitem{gc}
M.-M. Gu, R.-X. Hao, S.-M. Tang and J.-M. Chang, Analysis on component connectivity of bubble-sort star graphs and burnt pancake graphs, Discrete Appl. Math.  279 (2020), 
80--91, \doi{10.1016/j.dam.2019.10.018}.

\bibitem{gc22}
M.-M. Gu, J.-M. Chang,
Neighbor connectivity of pancake graphs and burnt pancake graphs,
Discrete Appl. Math. 324 (2022),
46--57, \doi{10.1016/j.dam.2022.09.013}.
 
\bibitem{Iwasaki-2010}
T. Iwasaki and K. Kaneko,
Fault-tolerant routing in burnt pancake graphs,
Inform. Process. Lett. 110 (2010), 535--538, 
\doi{10.1016/j.ipl.2010.04.023}.

\bibitem{hu}
X. Hu, H. Liu,
The (conditional) matching preclusion for burnt pancake graphs,
Discrete Appl. Math.  161 (2013), 1481--1489, 
\doi{10.1016/j.dam.2013.01.010}.

\bibitem{kk}
K. Kaneko,
Hamiltonian cycles and Hamiltonian paths in faulty burnt pancake graphs,
IEICE Trans. Inf. Syst. E90-D (4) (2007), 716--721, 
\doi{10.1093/ietisy/e90-d.4.716}.

\bibitem{Kim13}
S.-Y. Kim, J.-H. Park, 
Paired many-to-many disjoint path covers in recursive circulants $G(2m,4)$,
IEEE Trans. Comput. 62:12 (2013), 
2468--2475, 
\doi{10.1109/TC.2012.133}.

\bibitem{Lo}
L. Lov\'asz,
Combinatorial Problems and Exercises, North Holland, Amsterdam, 1993, 
\doi{10.1016/C2009-0-09109-0}.

\bibitem{Park22}
J.-H. Park,
Unpaired many-to-many disjoint path covers in nonbipartite torus-like graphs with faulty elements,
IEEE Access 10 (2022), 127589--127600, 
\doi{10.1109/ACCESS.2022.3226687}.

\bibitem{Park15}
J.-H. Park, I. Ihm, 
Many-to-many two-disjoint path covers in cylindrical and toroidal grids, 
Discrete Appl. Math. 185 (2015),
168--191, \doi{10.1016/j.dam.2014.12.008}.

\bibitem{Park06}
J.-H. Park, H.-C. Kim, H.-S. Lim,
Many-to-many disjoint path covers in hypercube-like interconnection networks with faulty elements,
IEEE Trans. Parallel Distrib. Syst. 17 (3) (2006),
227--240, 
\doi{10.1109/TPDS.2006.37}.

\bibitem{ParkLim22}
J.-H. Park, H.-S. Lim,
Characterization of interval graphs that are paired 2-disjoint path coverable,
J. Supercomput. (2022),
\doi{10.1007/s11227-022-04768-x}.

\bibitem{So15}
S.~Song S., X.~Li, S.~Zhou, M.~Chen,
Fault tolerance and diagnosability of burnt pancake networks under the comparison model,
Theoret. Comput. Sci. 582 (2015), 
48--59, \doi{10.1016/j.tcs.2015.03.027}.

\bibitem{So16}
S.~Song, S.~Zhou, X.~Li,
Conditional diagnosability of burnt pancake networks under the PMC model,
Comput. J. 59:1 (2016), 
91--105, 
\doi{10.1093/comjnl/bxv066}.

\end{thebibliography}
\end{document}